\documentclass[journal,twoside,10pt]{IEEEtran}

\usepackage{amsmath,amssymb,amsfonts,amsthm}
\usepackage{accents} 
\usepackage{algorithm}
\usepackage{algpseudocode}
\algrenewcommand\textproc{} 
\usepackage{balance}
\usepackage{bm}
\usepackage{booktabs}
\usepackage{cite}
\usepackage{enumitem}
\usepackage[T1]{fontenc}
\usepackage{graphicx}
\usepackage{mathrsfs}
\usepackage{mathtools}
\usepackage{multirow}
\usepackage{orcidlink}
\usepackage{siunitx}
\usepackage{soul}               
\usepackage{threeparttable}
\usepackage[color=green!40]{todonotes}
\soulregister\cite7
\soulregister\footnote7
\soulregister\eqref7
\soulregister\ref7
\usepackage{xcolor}
\usepackage{xstring} 
\usepackage{url}
\usepackage{hyperref}
\newlength{\myvspace}
\newtheorem{definition}{Definition}
\newtheorem{lemma}{Lemma}
\newtheorem{theorem}{Theorem}

\DeclareMathOperator*{\argmin}{arg\,min}

\DeclareMathOperator*{\fp}{fp}

\newcommand{\bigO}{\mathcal{O}}
\newcommand{\cpp}{C\texttt{++}}
\newcommand{\td}{\triangledown}

\newcommand{\mycase}[1]{\smallskip\noindent\textit{Case #1}} 

\graphicspath{{./}{figures/}{../figures/}}

\begingroup
\everyeof{\noexpand}
\xdef\normaljobname{\scantokens\expandafter{\jobname}}
\endgroup

\newcommand{\myauthor}{Kyeong Soo
  Kim$^{\orcidlink{0000-0002-4123-2647}}$,~\IEEEmembership{Senior~Member,~IEEE}}%
\newcommand{\mytitle}{Space-Time Trade-off in Integer Linear Scaling Rounded to
  the Nearest Integer through Multiplicative and Additive Decomposition}%

\begin{document}

\title{\mytitle}

\author{%
  \myauthor%
  \thanks{%
    This work was supported in part by the Postgraduate Research Scholarships
    (under Grant FOSA2412040 and PGRS1912001) of Xi'an Jiaotong-Liverpool
    University.

    K. S. Kim is with the Department of Communications and Networking, School of
    Advanced Technology, Xi'an Jiaotong-Liverpool University, Suzhou 215123,
    P. R. China (e-mail: Kyeongsoo.Kim@xjtlu.edu.cn).%
  }%
}%


\maketitle

\begin{abstract}
  We have been investigating clock skew compensation immune to floating-point
  precision loss by taking into account the discrete nature of clocks in digital
  communication systems, through which we have constructed more efficient
  incremental error algorithms using only integer operations by extending
  Bresenham's line drawing algorithm. In this paper, we formulate the problem of
  clock skew compensation as a special case of the integer linear scaling in the
  form of $i\frac{D}{A}$, where there is no offset and the scale factor is a
  product of two integers, and propose two algorithms---i.e., the
  \textit{multiplicative decomposition of integer division} and the
  \textit{additive decomposition of direct search}---for its nearest integer
  solution, which are not only immune to floating-point precision loss but also
  non-incremental unlike our prior approaches based on Bresenham's
  algorithm. Having theoretically established both decomposition algorithms
  based on a unified and rigorous formulation of the problem of the integer
  linear scaling rounded to the nearest integer, we discuss the space-time
  trade-off through the analysis of their computational complexities and
  non-overflow conditions. The numerical examples demonstrate the relative
  advantages and disadvantages of the two algorithms in a practical context of
  clock skew compensation under two different scenarios based on 32-bit and
  64-bit integers: We observe that the multiplicative decomposition of integer
  division algorithm can obtain the nearest integer solutions with the
  complexity of $\bigO(1)$ when $D$ is much smaller than the maximum value of
  the underlying integer type but overflows otherwise; in comparison, the
  additive decomposition of direct search algorithm can handle all the cases
  under both scenarios without overflows but at the expense of increased
  computational complexity when $i$ approaches the maximum value of the
  underlying integer type. We also observe that the additive decomposition of
  direct search algorithm based on 32-bit integers is equivalent to the clock
  skew compensation based on 64-bit double-precision floating-point arithmetic,
  while both algorithms based on 64-bit integers are equivalent to the clock
  skew compensation based on 128-bit quadruple-precision floating-point
  arithmetic, which highlights another trade-off between the bounded
  compensation errors and lower space complexity of the integer-based
  decomposition algorithms and the lower chances of overflows resulting from the
  wide ranges of numbers of the clock skew compensation based on floating-point
  arithmetic.
\end{abstract}

\begin{IEEEkeywords}
  Integer linear scaling, multiplicative decomposition, additive decomposition,
  non-incremental error algorithm, floating-point arithmetic, clock skew
  compensation.
\end{IEEEkeywords}

\section{Introduction}
\label{sec:introduction}
\IEEEPARstart{O}{serving} that the performance of clock skew compensation based
on floating-point arithmetic does not match with the prediction from theories
and simulations due to the use of 32-bit single-precision floating-point
arithmetic on resource-constrained platforms like wireless sensor network (WSN)
nodes and Internet of Things (IoT) devices, we have been investigating clock
skew compensation algorithms immune to floating-point precision loss: In
\cite{Kim:22-1}, we have proposed an incremental error algorithm based on the
extension of Bresenham's line drawing algorithm~\cite{bresenham65:_algor}, which
takes into account the discrete nature of clocks in digital communication
systems and thereby eliminates the effect of limited floating-point precision on
clock skew compensation. In \cite{Kang:24}, we have complemented the algorithm
by providing practical as well as theoretical bounds on the initial value of a
skew-compensated clock by floating-point division based on systematic analyses
of the errors of floating-point operations~\cite{jeannerod18}.

Note that the clock skew compensation algorithms immune to floating-point
precision loss studied in~\cite{Kim:22-1,Kang:24} are for a special case of
linear scaling, where the $y$-intercept is zero, the slope is given by a ratio
of two integers, and both domain and codomain are limited to integers. In fact,
the integer linear scaling problem has been studied for various applications,
including line drawing and pixel-perfect scaling in image/video processing and
visualization
~\cite{bresenham65:_algor,marti00:_memor,pitteway97,bostock11:_d3}, the
operation and data processing for the large area telescope (LAT) instrument in
the Gamma-ray large area space telescope (GLAST) satellite in
astronomy~\cite{lat_flight_sw}, and integer-only inference in machine
learning~\cite{liu25:_effic_detr,hu25:_tin_tin,wu20:_integ} in addition to clock
skew compensation in networking~\cite{Kim:20-1,Kim:22-1,Kang:23,Kang:24}, but
not based on a common framework and, oftentimes, without formal proofs and
systematically taking into account rounding errors. In this paper, we generalize
and formulate the problem of clock skew compensation as the \textit{integer
  linear scaling} problem in the form of $i\frac{D}{A}$ and propose two
decomposition algorithms that can avoid floating-point precision loss while
reducing overflows in obtaining the nearest integer solution based on a
fixed-width integer type.

The major contributions of our work can be summarized as follows: First, we
provide a unified and rigorous formulation of the problem of the integer linear
scaling rounded to the nearest integer. Second, based on the unified
formulation, we propose two decomposition algorithms---i.e., the
\textit{multiplicative decomposition of integer division} and the
\textit{additive decomposition of direct search}---for the nearest integer
solution to the integer linear scaling problem, which are not only immune to
floating-point precision loss but also non-incremental unlike our prior
approaches based on Bresenham's algorithm. Third, we discuss the space-time
trade-off in obtaining the nearest integer solution to the integer linear
scaling problem through the analysis of the computational complexities and the
non-overflow conditions of the proposed decomposition algorithms. Fourth, we
demonstrate the relative advantages and disadvantages of the two decomposition
algorithms through extensive numerical examples in a practical context of clock
skew compensation under two different scenarios based on 32-bit and 64-bit
integers, which also highlights another trade-off between the bounded
compensation errors and lower space complexity of the integer-based
decomposition algorithms and the lower chances of overflows resulting from the
wide ranges of numbers of the clock skew compensation based on floating-point
arithmetic.

The rest of the paper is organized as follows: In
Section~\ref{sec:csc-algo-review}, we review clock skew compensation algorithms
immune to floating-point precision loss with the underlying hardware clock
model. Section~\ref{sec:ils-rttn} describes the multiplicative decomposition of
integer division and the additive decomposition of direct search algorithms for
the integer linear scaling rounded to the nearest integer and discusses their
space-time trade-off through the analysis of the computational complexities and
the non-overflow conditions. Section~\ref{sec:numerical-examples} presents
numerical examples to further demonstrate the space-time trade-off in obtaining
the nearest integer solution to the integer linear scaling problem in the
practical context of clock skew compensation. We conclude our work in
Section~\ref{sec:conclusions}.

\section{Review of Clock Skew Compensation Algorithms Immune to Floating-Point
  Precision Loss}
\label{sec:csc-algo-review}
Time synchronization in networking provides a common time frame among network
nodes~\cite{wu11:_clock_synch_wirel_sensor_networ}, and we began our
investigation of the impact of limited precision floating-point arithmetic on it
at resource-constrained WSN/IoT nodes with the energy-efficient time
synchronization based on asynchronous source clock frequency recovery---EE-ASCFR
in short---scheme proposed in~\cite{Kim:17-1}; EE-ASCFR is based on the idea of
the separation of the clock frequency estimation/compensation at sensor nodes
and the clock offset and delay estimation at the head node and thereby can
reduce the complexity and power consumption of sensor nodes for time
synchronization by moving most of time synchronization operations to the head
node. The investigation of the practical implementation of EE-ASCFR on a real
WSN testbed composed of TelosB motes running TinyOS, however, reveals that its
performance on battery-powered, low-complexity sensor nodes is not up to that
predicted from theories and simulations due to the limited 32-bit
single-precision floating-point arithmetic of sensor nodes~\cite{Kim:20-1}.

Because we focus on the clock skew compensation at a sensor node in EE-ASCFR, we
consider the first-order affine clock model~\cite{chepuri13:_joint} describing
the hardware clock $T$ of the sensor node with respect to the reference clock
$t$ of the head node but without a clock offset as in~\cite{Kim:22-1,Kang:24}:
\begin{equation}
  \label{eq:hw_clock_model}
  T(t) = \left(1+\epsilon\right)t
\end{equation}
where $\epsilon{\in}\mathbb{R}$ is the clock skew. Compensating for the clock
skew from the hardware clock $T$ in \eqref{eq:hw_clock_model}, we can obtain the
logical clock $\hat{t}$ of the sensor node---i.e., the estimation of the
reference clock $t$ given the hardware clock $T$---as follows: For
$t_{i}{<}t{\leq}t_{i+1}~(i{=}0,1,{\ldots})$,
\begin{equation}
  \label{eq:logical_clock_model}
  \hat{t}\Big(T(t)\Big) = \hat{t}\Big(T(t_{i})\Big)
  + \dfrac{T(t)-T(t_{i})}{1 + \hat{\epsilon}_{i}},
\end{equation}
where $t_{i}$ is the reference time for the $i$th synchronization between the
head and the sensor node and $\hat{\epsilon}_{i}$ is the estimated clock skew
from the $i$th synchronization. The major issue is the floating-point division
required for the calculation of the second term of rhs in
\eqref{eq:logical_clock_model}, i.e.,
$\frac{T(t){-}T(t_{i})}{1{+}\hat{\epsilon}_{i}}$, which is the skew-compensated
increment of the hardware clock since the $i$th synchronization.

Because clocks in digital communication systems are basically discrete counters
and timestamps exchanged among nodes are their values, we proposed an
incremental error algorithm to obtain the second term of rhs in
\eqref{eq:logical_clock_model} using only integer addition/subtraction and
comparison by extending Bresenham's line drawing algorithm
in~\cite{Kim:22-1,Kang:24}: If $\frac{D}{A}$ is the inverse of a clock frequency
ratio (i.e., $\frac{1}{1{+}\epsilon_{i}}$) estimated based on two positive
integers $D$ and $A$, where $D$ and $A$ represent interdeparture and
interarrival times of packets or their cumulative sums from the previous
synchronization~\cite{Kim:13-1}, the clock skew can be compensated for based on
Theorem~1:
\begin{theorem}[Clock Skew Compensation Based on the Extended Bresenham's
  Algorithm~\cite{Kim:22-1,Kang:24}]
  \label{thm:csc_with_optimal_bounds}
  Given a hardware clock $i$, we can obtain its skew-compensated clock $j$ as
  follows:
	
  \smallskip%
  \noindent
  \textit{Case 1.} $\frac{D}{A}{<}1$: The skew-compensated clock $j$ satisfies
  \begin{equation}
    \label{eq:csc-scc_bounds}
    \left\lfloor{i\frac{D}{A}}\right\rfloor \leq j \leq \left\lceil{i\frac{D}{A}}\right\rceil.
  \end{equation}
  Unless $i\frac{D}{A}$ is an integer, there are two values satisfying
  \eqref{eq:csc-scc_bounds}. Because we cannot know the exact value of
  $i\frac{D}{A}$ due to limited floating-point precision, however, we extend
  \eqref{eq:csc-scc_bounds} to include the effect of the precision loss: For
  floating-point numbers with a base 2 and a precision $p$\footnote{$p$ is
    precision in bits; for example, $p{=}24$ for the 32-bit single-precision
    floating-point format (i.e., binary32) defined in IEEE
    754-2008~\cite{IEEE:754-2008}.},
  \begin{equation}
    \label{eq:csc-scc_optimal_bounds}
    \left\lfloor{\frac{1{-}u{+}2u^{2}}{(1{+}u)^{2}(1{+}2u)}t}\right\rfloor \leq j \leq
    \left\lceil{\frac{(1+2u)^{3}(1{+}u{-}2u^{2})}{(1{+}u)^{2}}t}\right\rceil,
  \end{equation}
  where $t{=}i\frac{D}{A}$ and $u{=}2^{-p}$.
	
  Let $k,{\ldots},k{+}l$ be the candidate values of $j$ satisfying
  \eqref{eq:csc-scc_optimal_bounds}. We determine $j$ by applying the extended
  Bresenham's algorithm of~\cite{Kim:22-1} with $\Delta{a}$ and $\Delta{b}$ set
  to $A$ and $D$ from the point $(i{-}l,k)$ and on; $j$ is determined by the $y$
  coordinate of the valid point whose $x$ coordinate is $i$.
	
  \smallskip%
  \noindent
  \textit{Case 2.} $\frac{D}{A}{>}1$: In this case, we can decompose the
  skew-compensated clock $j$ into two components as follows:
  \begin{equation}
    \label{eq:decomposition}
    j = i\frac{D}{A} = i + i\frac{D - A}{A}.
  \end{equation}
  Now that $\frac{D{-}A}{A}{<}1$, we can apply the same procedure of Case~1 to
  the second component in \eqref{eq:decomposition} by setting $\Delta{a}$ and
  $\Delta{b}$ of the extended Bresenham's algorithm to $A$ and $D{-}A$,
  respectively. Let $\bar{j}$ be the result from the procedure. The
  skew-compensated clock $j$ is given by $i{+}\bar{j}$ as per
  \eqref{eq:decomposition}.
\end{theorem}
\noindent%
Note that practical bounds loosening the theoretical bounds in
\eqref{eq:csc-scc_optimal_bounds} are also provided for practical implementation
at resource-constrained sensor nodes with limited floating-point precision
in~\cite{Kang:24}.

As illustrated in Fig.~\ref{fig:csc-ext-bresenham}~(b), the range of the initial
value of $j$, which are bounded in \eqref{eq:csc-scc_optimal_bounds}, determines
the number of iterations of the algorithm (i.e., the difference between the
upper and lower bounds).
\begin{figure}[!tb]
  \begin{center}
    \includegraphics[angle=-90,width=.9\linewidth]{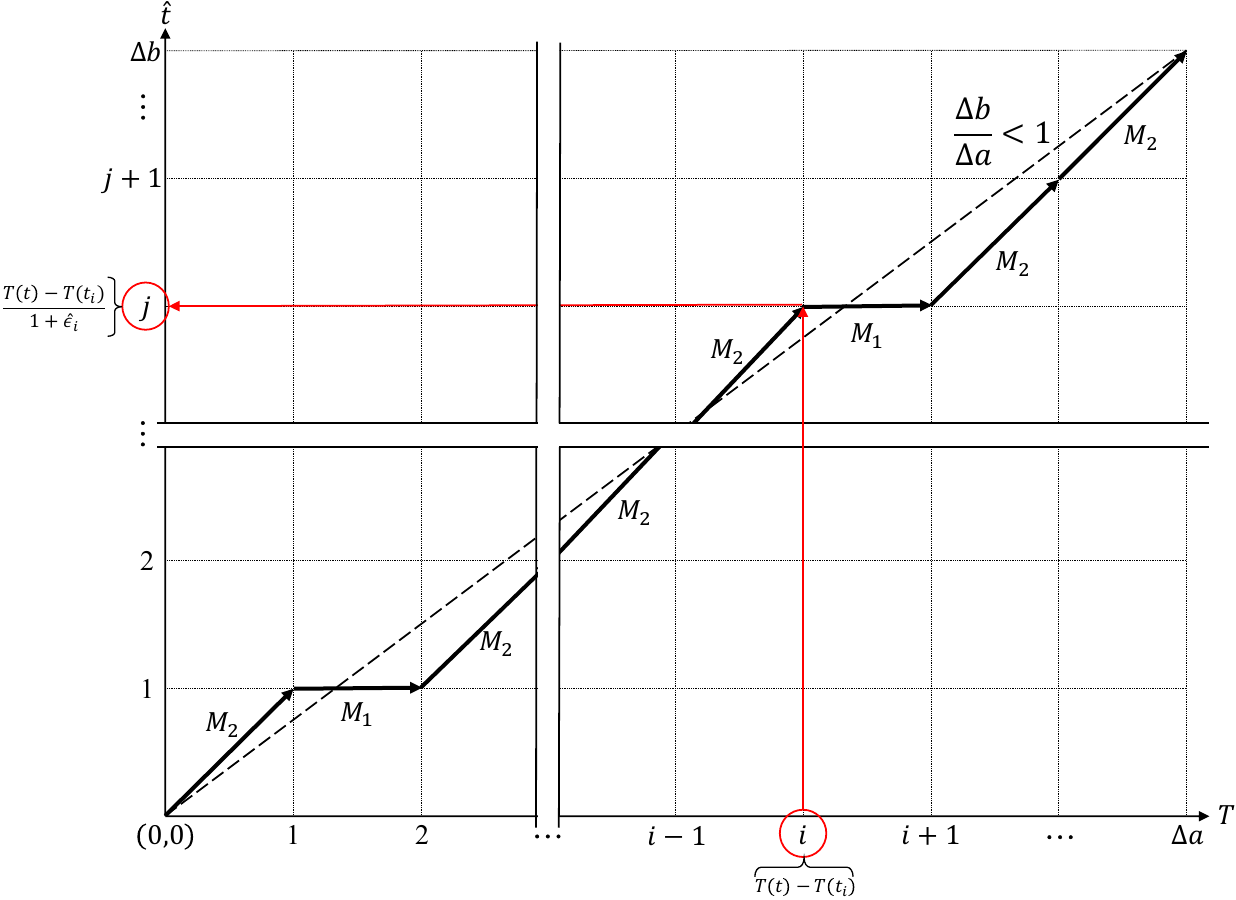}\\
    {\scriptsize (a)}\\
    \includegraphics[angle=-90,width=.95\linewidth]{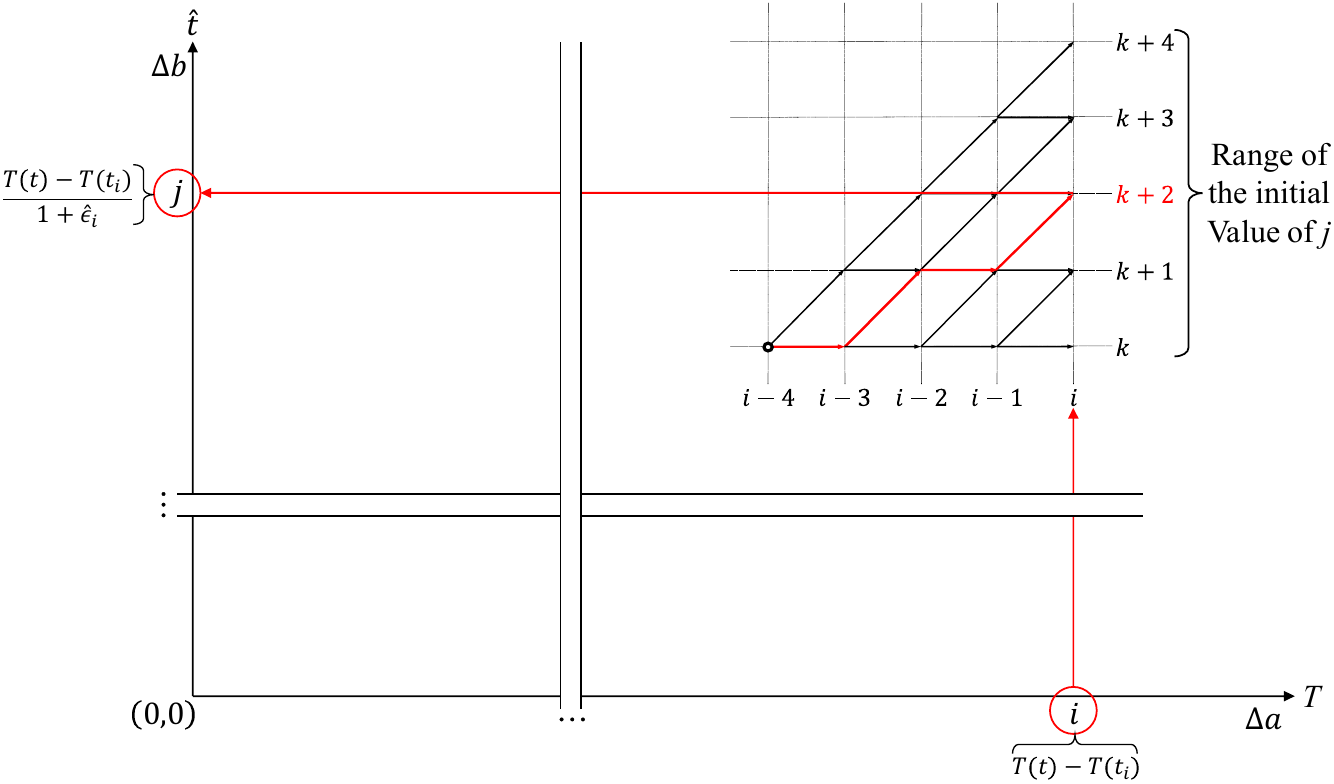}\\
    {\scriptsize (b)}
  \end{center}
  \caption{Clock skew compensation based on (a) the
    original~\cite{bresenham65:_algor} and (b) the extended Bresenham's
    algorithm~\cite{Kim:22-1,Kang:24} for the case of $\frac{D}{A}{<}1$.}
  \label{fig:csc-ext-bresenham}
\end{figure}
Though the clock skew compensation based on the extended Bresenham's algorithm
can greatly reduce the number of iterations compared to that based on the
original algorithm as shown in Fig.~\ref{fig:csc-ext-bresenham}, the numerical
examples presented in~\cite{Kang:24} demonstrate that the theoretical and
practical bounds based on the systematic analyses of floating-point operation
errors are looser than the approximate bounds of~\cite{Kim:22-1} and still
increase the number of iterations, which is the motivation for the current work
on non-incremental error algorithms not relying on the bounds on the initial
value of the clock.

\section{Integer Linear Scaling Rounded to the Nearest Integer}
\label{sec:ils-rttn}
Without loss of generality, we confine our discussions to the special case of
integer linear scaling with two non-negative integers $i$ and $D$ and a positive
integer $A$; the nearest integer solution to a general case can be obtained by
applying the sign of an original problem to the solution to its non-negative
version. For convenience, we also assume $i{\geq}D$ unless stated otherwise.

\subsection{Preliminaries}
\label{sec:preliminaries}
Let $\mathbb{N}_{0}$ and $\mathbb{N}_{+}$ denote the set of non-negative
integers and the set of positive integers, respectively. The problem of the
clock skew compensation can be generalized to the problem of integer linear
scaling rounded to the nearest integer as follows:
\begin{definition}
  \label{def:ils-rttn}
  Given $i,D{\in}\mathbb{N}_{0}$ and $A{\in}\mathbb{N}_{+}$, we define the
  nearest integer solution $j$ to $i\frac{D}{A}$ as follows:
  \begin{equation}
    \label{eq:ils-rttn-1}
    j \triangleq \argmin_{k \in \mathbb{N}_{0}}\left|k-i\frac{D}{A}\right|.
  \end{equation}
\end{definition}

Due to the lack of tie-breaking, Definition~\ref{def:ils-rttn} allows up to two
nearest integer solutions (e.g., 7 or 8 for
$3\frac{5}{2}{=}7.5$). Lemma~\ref{lem:ils-rttn} provides an equivalent
condition for the nearest integer solution without division.
\begin{lemma}
  \label{lem:ils-rttn}
  Given $i,D{\in}\mathbb{N}_{0}$ and $A{\in}\mathbb{N}_{+}$, $j$ is the nearest
  integer solution to $i\frac{D}{A}$ if and only if $j$ satisfies the following
  condition:
  \begin{equation}
    \label{eq:ils-rttn-2}
    \left|jA-iD\right| \leq \left|kA-iD\right| ~~~ \forall k \in \mathbb{N}_{0}.
  \end{equation}
\end{lemma}
\begin{IEEEproof}
  If $j$ is the nearest integer solution to $i\frac{D}{A}$, $j$ satisfies the
  following condition by Definition~\ref{def:ils-rttn}:
  \begin{equation}
    \label{eq:ils-rttn-3}
    \left|j-i\frac{D}{A}\right| \leq \left|k-i\frac{D}{A}\right| ~~~ \forall k \in \mathbb{N}_{0}.
  \end{equation}
  Multiplying both sides of \eqref{eq:ils-rttn-3} by $A$, we obtain
  \eqref{eq:ils-rttn-2}.
  
  If $j$ satisfies \eqref{eq:ils-rttn-2}, we obtain \eqref{eq:ils-rttn-3} by
  dividing both sides of \eqref{eq:ils-rttn-2} by $A$, which means that $j$ is
  the nearest integer solution to $i\frac{D}{A}$ according to
  Definition~\ref{def:ils-rttn}.
\end{IEEEproof}
\vspace{\myvspace}%

Lemma~\ref{lem:ils-rttn-fp} provides the nearest integer solution based on the
exact value of $i\frac{D}{A}$.
\begin{lemma}
  \label{lem:ils-rttn-fp}
  Given $i,D{\in}\mathbb{N}_{0}$ and $A{\in}\mathbb{N}_{+}$, we can obtain the
  nearest integer solution $j$ to $i\frac{D}{A}$ as follows:
  \begin{equation}
    \label{eq:ils-rttn-fp}
    j = \left\lfloor i \frac{D}{A} + 0.5 \right\rfloor.
  \end{equation}
\end{lemma}
\begin{IEEEproof}
  See Appendix~\ref{sec:lem-ils-rttn-fp-proof}.
\end{IEEEproof}
\vspace{\myvspace}%

Unlike Definition~\ref{def:ils-rttn}, as Lemma~\ref{lem:ils-rttn-fp}
incorporates the \textit{round half up} tie-breaking rule, it results in a
unique solution. A major issue with Lemma~\ref{lem:ils-rttn-fp} is that we
cannot obtain the exact value of $i\frac{D}{A}$ on a real platform but only an
approximate value of it due to floating-point precision loss, i.e.,
${\fp}(i\frac{D}{A})$ where $\fp(\cdot)$ denotes the result of the
floating-point operations on the given platform.

From Lemmas~\ref{lem:ils-rttn} and \ref{lem:ils-rttn-fp}, we can derive a bound
on the lhs of \eqref{eq:ils-rttn-2}:
\begin{equation}
  \label{eq:td-bound}
  \left|jA-iD\right| \leq \frac{A}{2}.
\end{equation}

\subsection{Algorithms Immune to Floating-Point Precision Loss}
\label{sec:ils-rttn-algos}
We present efficient algorithms for the integer linear scaling rounded to the
nearest integer, which are immune to floating-point precision loss and reduce
overflows in obtaining the nearest integer solution based on a fixed-width
integer type through multiplicative and additive decomposition techniques.

\subsubsection{Multiplicative Decomposition of Integer Division}
\label{sec:mdid}
Theorem~\ref{thm:mdid} presents an algorithm based on multiplicative
decomposition of integer division using the fixed-point
representation~\cite{dsl:fixed-point-arithmetic}, not relying on floating-point
arithmetic and thereby immune to floating-point precision loss.
\begin{theorem}[Multiplicative decomposition of integer division]
  \label{thm:mdid}
  Given $i,D{\in}\mathbb{N}_{0}$ and $A{\in}\mathbb{N}_{+}$, we can obtain the
  nearest integer solution $j$ to $i\frac{D}{A}$ as follows:
  \begin{equation}
    \label{eq:mdid}
    j = \left\lfloor \frac{i}{A} \right\rfloor D +
    \left\lfloor
      \frac{\left(i \bmod A\right)D + \left\lfloor\frac{A}{2}\right\rfloor}{A}
    \right\rfloor.
  \end{equation}
\end{theorem}
\begin{IEEEproof}
  See Appendix~\ref{sec:thm-mdid-proof}.
\end{IEEEproof}
\vspace{\myvspace}%

Unlike the approach based on normal integer division (i.e., $\frac{iD}{A}$),
Theorem~\ref{thm:mdid} enables us to obtain the nearest integer solution to
$i\frac{D}{A}$ without calculating the product $iD$ for the dividend, which
could overflow when the result cannot be represented by a fixed-width integer
type.
If both $i$ and $D$ are smaller than $A$, however,
\eqref{eq:mdid} reduces back to normal integer division rounded to the nearest
integer, i.e.,
\begin{equation}
  \label{eq:fd}
  j = 
  \left\lfloor
    \frac{iD + \left\lfloor\frac{A}{2}\right\rfloor}{A}
  \right\rfloor.
\end{equation}

\subsubsection{Additive Decomposition of Direct Search}
\label{sec:adds}
Noting that finding the nearest integer solution to $i\frac{D}{A}$ is an integer
programming problem whose objective function is $f(x){=}|x{-}i\frac{D}{A}|$
(i.e., \eqref{eq:ils-rttn-1} of Definition~\ref{def:ils-rttn}), we provide an
alternative algorithm directly searching the minimum in Theorem~\ref{thm:ds}.
\begin{theorem}[Direct search]
  \label{thm:ds}
  Given $i,D{\in}\mathbb{N}_{0}$ and $A{\in}\mathbb{N}_{+}$,
  we can obtain $j$, the nearest integer solution to $i\frac{D}{A}$, and
  $\Delta$, the value of $jA{-}iD$, using Algorithm~\ref{alg:ds}\footnote{The
    internal variables of $k_{0}$, $k_{1}$, $\Delta_{0}$, and $\Delta_{1}$ are
    subscripted to differentiate their values between the updates during the
    proof.} as follows:
  \begin{equation}
    \label{eq:ds}
    (j, \Delta) = ds(i, D, A, \kappa, 0),
  \end{equation}
  where $\kappa{\in}\mathbb{N}_{0}$ is an initial guess of $j$.
\end{theorem}
\begin{algorithm}[!tbh]
  \caption{Inter linear scaling to the nearest integer by direct search.}
  \label{alg:ds}
  \begin{algorithmic}[1]
    \Function{ds}{$i$, $D$, $A$, $k_{0}$, $\Delta_{\!-}$}
      \State $\Delta_{0} \Leftarrow (k_{0} - i)A + i(A - D) + \Delta_{\!-}$%
      \If{$\Delta_{0} = 0$} \Comment{\textit{Case}~1}%
        \State $j \Leftarrow k_{0}$%
      \ElsIf{$\Delta_{0} > 0$} \Comment{\textit{Case}~2}%
        \State $k_{1} \Leftarrow k_{0} - \left\lfloor\dfrac{\Delta_{0}}{A}\right\rfloor$%
        \vspace{0.4ex}%
        \State $\Delta_{1} \Leftarrow \Delta_{0} \bmod A$%
        \If{$|\Delta_{1} - A| < |\Delta_{1}|$} \Comment{\textit{Case}~2.1}%
          \State $j \Leftarrow k_{1} - 1$%
          \State $\Delta \Leftarrow \Delta_{1} - A$%
        \Else \Comment{\textit{Case}~2.2}%
          \State $j \Leftarrow k_{1}$%
          \State $\Delta \Leftarrow \Delta_{1}$%
        \EndIf%
      \Else \Comment{\textit{Case}~3}%
        \State $k_{1} \Leftarrow k_{0} - \left\lfloor\dfrac{\Delta_{0}}{A}\right\rfloor$%
        \vspace{0.35ex}%
        \State $\Delta_{1} \Leftarrow - \left(|\Delta_{0}| \bmod A\right)$%
        \If{$\left|\Delta_{1} + A\right| < |\Delta_{1}|$} \Comment{\textit{Case}~3.1}%
          \State $j \Leftarrow k_{1} + 1$%
          \State $\Delta \Leftarrow \Delta_{1} + A$%
        \Else \Comment{\textit{Case}~3.2}%
          \State $j \Leftarrow k_{1}$%
          \State $\Delta \Leftarrow \Delta_{1}$%
        \EndIf%
      \EndIf%
      \State \Return $j$, $\Delta$%
    \EndFunction%
  \end{algorithmic}
\end{algorithm}
\begin{IEEEproof}
  See Appendix~\ref{sec:thm-ds-proof}.
\end{IEEEproof}
\vspace{\myvspace}%

Though $j$ from Algorithm~\ref{alg:ds} does not depend on the value of $k_{0}$
(i.e., $\kappa$ in \eqref{eq:ds}), it can be set to
$\lfloor\fp(i\frac{D}{A}){+}0.5\rfloor$ if floating-point arithmetic is
available; otherwise, it can be set to ${\lfloor}\frac{i}{A}{\rfloor}D$,
${\lfloor}\frac{D}{A}{\rfloor}i$, or just $i$ depending on the values of $i$,
$D$, and $A$. Also, the initialization of $\Delta$ in Line~2 is to avoid
overflow during the multiplication of $kA$ or $iD$ because
\begin{equation*}
  \Delta_{0} = k_{0}A - iD = (k_{0}-i)A + i(A-D),
\end{equation*}
especially when $\frac{D}{A}{\approx}1$, which is the case for clock skew
compensation.




When the value of $i$ is very large, there could be a chance of an overflow
during the initialization of $\Delta$. To avoid the overflow, we can divide the
given problem with $i$ into a series of sub-problems with smaller values of
$i_{n}$ ($n{=}1,{\ldots},N$), where $\sum_{i=1}^{N}i_{n}{=}i$. The main
challenge is that the rounding error from a subproblem with $i_{n}$ needs to be
properly taken into account in the next subproblem with $i_{n+1}$ so that there
should be no cumulation of rounding errors in the end.

When we decompose a problem of integer linear scaling to the nearest integer by
direct search into two subproblems as shown in Fig.~\ref{fig:csc-decomp}, the
second subproblem starts from $(0, \delta_{1})$, not the origin, due to the
rounding error of the first subproblem.
\begin{figure*}[!tb]
  \begin{center}
    \includegraphics[width=\linewidth]{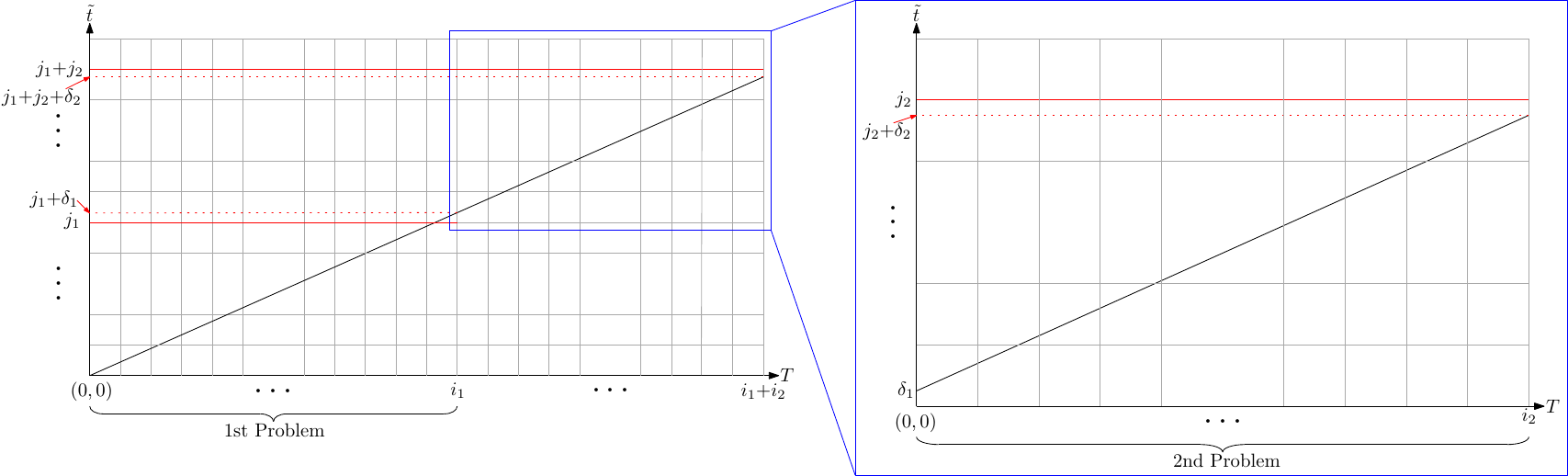}
  \end{center}
  \caption{Decomposition of integer linear scaling to the nearest integer by
    direct search into two subproblems.}
  \label{fig:csc-decomp}
\end{figure*}
$\Delta$ resulting from Algorithm~\ref{alg:ds}, which indicates how close $j$ is
to $i\frac{D}{A}$ without division, can be used to carry out the information on
the rounding error from one subproblem to another as stated in
Lemma~\ref{lem:ds-dcmp}.
\begin{lemma}
  \label{lem:ds-dcmp}
  Given $i_{1},i_{2},D{\in}\mathbb{N}_{0}$ and $A{\in}\mathbb{N}_{+}$, we can
  obtain $j$, the nearest integer solution to $i\frac{D}{A}$, where
  $i{=}i_{1}{+}i_{2}$, and $\Delta$, the value of $jA{-}iD$, using
  Algorithm~\ref{alg:ds} as follows:
  \begin{equation}
    \label{eq:lem-ds-dcmp-1}
    \begin{split}
      j & = j_{1} + j_{2}, \\
      \Delta & = \tilde{\Delta}_{2},
    \end{split}
  \end{equation}
  where
  \begin{equation}
    \label{eq:lem-ds-dcmp-2}
    \begin{split}
      (j_{1}, \tilde{\Delta}_{1}) & = ds(i_{1}, D, A, \kappa_{1}, 0), \\
      (j_{2}, \tilde{\Delta}_{2}) & = ds(i_{2}, D, A, \kappa_{2}, \tilde{\Delta}_{1}),
    \end{split}
  \end{equation}
  and $\kappa_{n}{\in}\mathbb{N}_{0}$ ($n{=}1,2$) is an initial guess of
  $j_{n}$.
\end{lemma}
\begin{IEEEproof}
  See Appendix~\ref{sec:lem-ds-dcmp-proof}.
\end{IEEEproof}
\vspace{\myvspace}%
\noindent%

Theorem~\ref{thm:ds-dcmp} extends Lemma~\ref{lem:ds-dcmp} to a general case.
\begin{theorem}[Additive decomposition of direct search]
  \label{thm:ds-dcmp}
  Given $i_{1},{\ldots},i_{N},D{\in}\mathbb{N}_{0}$ and $A{\in}\mathbb{N}_{+}$,
  we can obtain $j$, the nearest integer solution to $i\frac{D}{A}$, where
  $i{=}\sum_{n=1}^{N}i_{n}$, and $\Delta$, the value of $jA{-}iD$, using
  Algorithm~\ref{alg:ds} for every integer $N{\geq}1$ as follows:
  \begin{equation}
    \label{eq:thm-ds-dcmp-1}
    \begin{split}
      j & = \sum_{n=1}^{N}j_{n}, \\
      \Delta & = \tilde{\Delta}_{N},
    \end{split}
  \end{equation}
  where for $n{=}1,{\ldots},N$,
  \begin{equation}
    \label{eq:thm-ds-dcmp-2}
    \begin{split}
      (j_{n}, \tilde{\Delta}_{n}) & = ds(i_{n}, D, A, \kappa_{n}, \tilde{\Delta}_{n-1}), \\
      \tilde{\Delta}_{0} & = 0,
    \end{split}
  \end{equation}
  and $\kappa_{n}{\in}\mathbb{N}_{0}$ is an initial guess of $j_{n}$.
\end{theorem}

\begin{IEEEproof}
  See Appendix~\ref{sec:thm-ds-dcmp-proof}.
\end{IEEEproof}
\vspace{\myvspace}%

Theorem~\ref{thm:ds-dcmp} enables us to decompose the original problem of
integer linear scaling rounded to the nearest integer with a large value of
$i{=}\sum_{n=1}^{N}i_{n}$ into subproblems with its smaller components (i.e.,
$i_{n}$) without incurring additional rounding errors during the decomposition
and the summation. It is $\tilde{\Delta}_{l}$ ($l{<}N$) that carries the
information on the rounding error of the nearest integer solution to
$\sum_{n=1}^{l}i_{n}\frac{D}{A}$ to the $(l{+}1)$th subproblem; as such,
$\tilde{\Delta}_{l}$ is sequential, yet its absolute value is bounded by
$\frac{A}{2}$ (i.e., \eqref{eq:td-bound}).

Fig.~\ref{fig:adds-ana} shows the integer linear scaling errors resulting from
setting $\tilde{\Delta}_{n}$ to zero for $n{=}1,{\ldots},N{-}1$ in
\eqref{eq:thm-ds-dcmp-2} of the additive decomposition of direct search, where
the integer linear scaling errors are calculated with respect to
${\lfloor}\fp(i\frac{D}{A}){+}0.5{\rfloor}$ based on binary512 of IEEE
754-2008~\cite[Table~III]{Kang:24}.
\begin{figure}[!tb]
  \begin{center}
    \includegraphics[width=.9\linewidth,trim={0 10 0 15},clip]{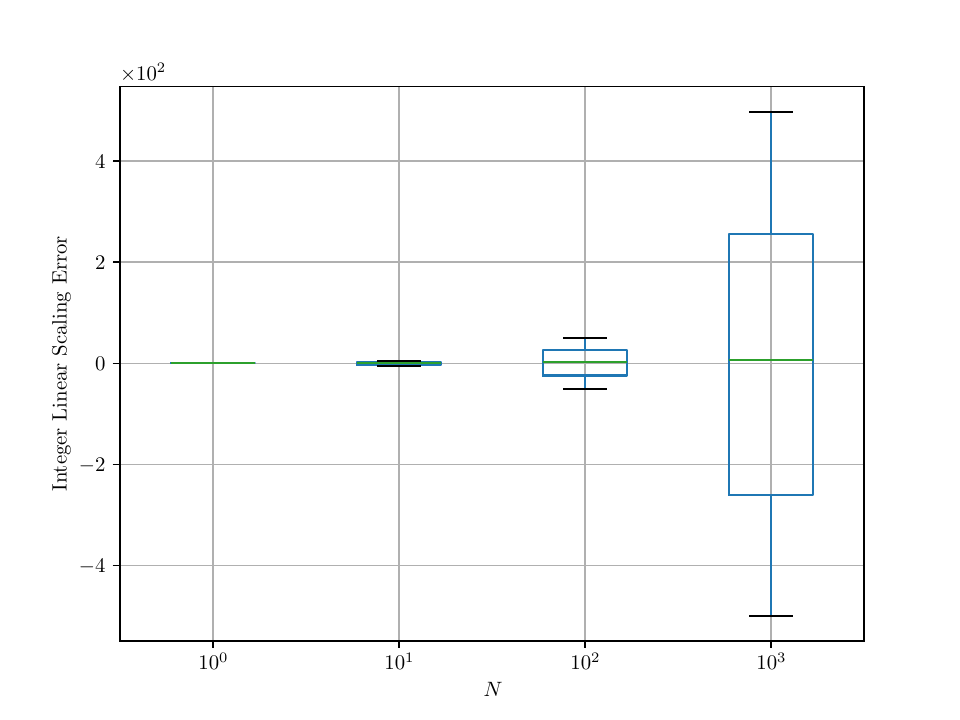}\\
    {\scriptsize (a)}\\
    \includegraphics[width=.9\linewidth,trim={0 10 0 15},clip]{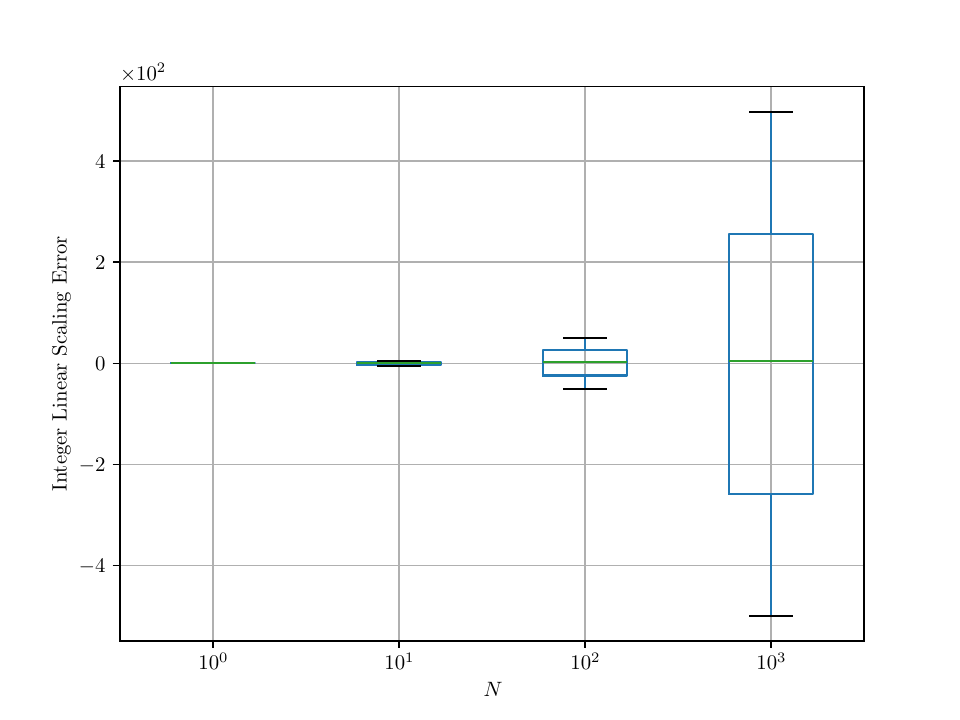}\\
    {\scriptsize (b)}
  \end{center}
  \caption{Integer linear scaling errors with $\tilde{\Delta}_{n}{=}0$ for
    $n{=}1,{\ldots},N{-}1$ in \eqref{eq:thm-ds-dcmp-2} of the additive
    decomposition of direct search based on (a) 32-bit and (b) 64-bit integers.}
  \label{fig:adds-ana}
\end{figure}
We set $i$ and $D$ to \num{1e6} for 32-bit integers and \num{1e12} for 64-bit
integers, respectively. We also set $A$ to $D{+}u$, where $u$ is an integer
uniformly distributed in the range of $[1,\num{1e3}]$ for 32-bit integers and
$[1,\num{1e6}]$ for 64-bit integers, respectively. Because $i{<}A$, $D{<}A$, and
$iA{\geq}2^{W\!-1}\!{-}1$, not only normal integer division but also the
multiplicative decomposition of integer division cannot handle these cases due
to overflows. From the results, we can observe that the integer linear scaling
errors are bounded by ${\pm}0.5N$ as the rounding errors are bounded by
${\pm}0.5$, which demonstrate how critical the transfer of $\tilde{\Delta}_{n}$
between the consecutive sub-problems in the additive decomposition of direct
search is.

\subsection{Space-Time Trade-off}
\label{sec:st-tradeoff}
We analyze the computational complexities and the non-overflow conditions of the
algorithms presented in Theorems~\ref{thm:mdid}--\ref{thm:ds-dcmp} and
Bresenham's algorithm~\cite{bresenham65:_algor} for integer linear scaling
rounded to the nearest integer, whose results are summarized in
Table~\ref{tab:algo-summary}.
\begin{table*}[!tb]
  \centering
  \caption{Summary of the algorithms for integer linear scaling rounded to the
    nearest integer.}
  \label{tab:algo-summary} {%
    \renewcommand{\arraystretch}{1.3}%
    \begin{tabular}{lccl} \toprule%
      \multicolumn{1}{c}{\multirow{2}{*}{Algorithm}} & \multicolumn{2}{c}{Complexity} & \multicolumn{1}{c}{\multirow{2}{*}{Non-overflow conditions}} \\
      \cmidrule{2-3}%
                                                     & Space & Time & \\
      \midrule%
      Bresenham's algorithm~\cite{bresenham65:_algor} & $\bigO(1)$ & $\bigO\left(\min(i, D)\right)$ & {\renewcommand{\arraystretch}{1.3}\begin{tabular}{@{}l@{}} $i{\leq}2^{W\!-1}\!{-}1$, $D{\leq}2^{W\!-1}\!{-}1$, $A{\leq}2^{W\!-1}\!{-}1$,\\ $\left\lfloor i\frac{D}{A}{+}0.5 \right\rfloor \leq 2^{W\!-1}\!{-}1$, and\\ ${-}2^{W\!-2} \leq A \leq 2^{W\!-2}\!{-}1$. \end{tabular}} \\
      \addlinespace%
      {\renewcommand{\arraystretch}{1.3}\begin{tabular}{@{}l@{}}Multiplicative decomposition of\\integer division (Theorem~\ref{thm:mdid})\end{tabular}} & $\bigO(1)$ & $\bigO(1)$ & {\renewcommand{\arraystretch}{1.3}\begin{tabular}{@{}l@{}} $i{\leq}2^{W\!-1}\!{-}1$, $D{\leq}2^{W\!-1}\!{-}1$, $A{\leq}2^{W\!-1}\!{-}1$,\\ $\left\lfloor i\frac{D}{A}{+}0.5 \right\rfloor \leq 2^{W\!-1}\!{-}1$,\\ ${\lfloor}\frac{i}{A}{\rfloor}D{\leq}2^{W\!-1}\!{-}1$, $(i{\bmod}A)D{+}{\lfloor}\frac{A}{2}{\rfloor}{\leq}2^{W\!-1}\!{-}1$, and\\ $iD \leq 2^{W\!-1}\!{-}1$ when $i{<}A$ and $D{<}A$. \end{tabular}} \\
      \addlinespace%
      Direct search (Theorem~\ref{thm:ds}) & $\bigO(1)$ & $\bigO(1)$ & {\renewcommand{\arraystretch}{1.3}\begin{tabular}{@{}l@{}} $i{\leq}2^{W\!-1}\!{-}1$, $D{\leq}2^{W\!-1}\!{-}1$, $A{\leq}2^{W\!-1}\!{-}1$,\\ $\left\lfloor i\frac{D}{A}{+}0.5 \right\rfloor \leq 2^{W\!-1}\!{-}1$, and\\
        ${-}2^{W\!-1} \leq (\kappa{-}i)A{+}i(A{-}D) \leq 2^{W\!-1}{-}1$. \end{tabular}} \\
      \addlinespace%
      {\renewcommand{\arraystretch}{1.3}\begin{tabular}{@{}l@{}}Additive decomposition of\\direct search (Theorem~\ref{thm:ds-dcmp})\end{tabular}} & $\bigO(N)$ & $\bigO(N)$ & {\renewcommand{\arraystretch}{1.3}\begin{tabular}{@{}l@{}} For $n{=}1,{\ldots},N$,\\ $i_{n}{\leq}2^{W\!-1}\!{-}1$, $D{\leq}2^{W\!-1}\!{-}1$, $A{\leq}2^{W\!-1}\!{-}1$,\\ $\left\lfloor \left(\sum_{i=1}^{N}i_{n}\right)\frac{D}{A}{+}0.5 \right\rfloor \leq 2^{W\!-1}\!{-}1$, and\\
        ${-}2^{W\!-1}{+}\frac{A}{2} \leq (\kappa_{n}{-}i_{n})A{+}i_{n}(A{-}D) \leq 2^{W\!-1}{-}1{-}\frac{A}{2}$. \end{tabular}} \\
      \bottomrule%
    \end{tabular}
  }%
\end{table*}
As stated in the beginning of Section~\ref{sec:ils-rttn}, the nearest integer
solution to a general case can be obtained by applying the sign of an original
problem to the solution to its non-negative version provided by the algorithms
under consideration, so we assume that not only the input and output values but
also all the intermediate values of an algorithm are represented by a
fixed-width $W$-bit signed integer type (e.g., \texttt{int32\_t} in C/\cpp) for
the analysis.

Because the space complexity captured by $\bigO({\cdot})$ describes only the
limiting behavior of the memory requirements, we also investigate the
non-overflow conditions of each algorithm during its entire operation, which
gives us an alternative and practical measure of its space complexity; when
Algorithm A needs a bigger-size integer type than Algorithm B to avoid
overflows, we could say that the space complexity of Algorithm A is higher than
Algorithm B.


As for Bresenham's algorithm, due to its recursive construction of
$\td_{(\cdot)}$ whose absolute value is bounded by
$2A$~\cite[Lemma~1]{bresenham65:_algor}, we have the third condition for its
representation by a $W$-bit signed integer type. The direct search algorithm in
Theorem~\ref{thm:ds} and its additive decomposition in
Theorem~\ref{thm:ds-dcmp}, too, use $\Delta_{(\cdot)}$ whose absolute value is
bounded by $\frac{A}{2}$ and therefore require the third condition for
$\Delta_{0}$. The third conditions of the multiplicative decomposition of
integer division in Theorem~\ref{thm:mdid}, on the other hand, are for the
intermediate calculations in \eqref{eq:mdid}; its fourth condition for the
product $iD$ is for the case when both $i$ and $D$ are smaller than $A$, which
reduces it back to normal integers division rounded to the nearest integer as
discussed in Section~\ref{sec:mdid}.

Regarding the implications of the non-overflow conditions, we first consider a
special case of $iD{=}2^{W\!-1}$ and $A{=}2^{W\!-1}\!{-}1$ for the
multiplicative decomposition of integer division algorithm. The nearest integer
solution to $i\frac{D}{A}$ in this case is just one for $W{>}1$, but the product
$iD$ results in overflow when $i{\geq}2$ or $D{\geq}2$.
The direct search algorithm can handle this case by setting $i{=}2$,
$D{=}2^{W\!-2}$, and $k_{0}{=}i$ without overflows. Bresenham's algorithm,
however, cannot support $A{=}2^{W\!-1}\!{-}1$ due to the condition on $A$.

For another special case of $\frac{D}{A}{\approx}1$, which holds for clock skew
compensation, the major issue is the non-overflow condition on $i$ as we need
the nearest integer solution to $i\frac{D}{A}$ for a large value of
$i$. Compared to Bresenham's and the multiplicative decomposition of integer
division algorithms, the direct search algorithm can provide the nearest integer
solution even for $i{\geq}2^{W\!-1}\!{-}1$ thanks to the additive decomposition
at the expense of the increased time complexity.

The results of the analysis of the computational complexities and non-overflow
conditions of the algorithms summarized in Table~\ref{tab:algo-summary} and the
related discussions on their implications show that the direct search algorithm
can better exploit the trade-off between space and time complexity in obtaining
the nearest integer solution to integer linear scaling through additive
decomposition together with flexible setting of an initial guess of the solution
(i.e., $k_{0}$ in Algorithm~\ref{alg:ds}).

\section{Numerical Examples: Clock Skew Compensation}
\label{sec:numerical-examples}
To demonstrate the space-time trade-off in integer linear scaling rounded to the
nearest integer in a practical context, we apply the decomposition algorithms
described in Section~\ref{sec:ils-rttn-algos} to clock skew compensation under
two different scenarios: The first scenario based on 32-bit integers assumes a
resource-constrained platform for WSN/IoT running on 1-\si{\MHz} clock and
focuses on the capability of each algorithm carrying out clock skew compensation
without overflows. The second scenario based on 64-bit integers, on the other
hand, targets a resourceful platform running on 1-\si{\GHz} clock, where the
focus is shifted to the scalability of the algorithms even with extremely large
values of $i$ up to \num{1e18}.

The results of clock skew compensation based on the decomposition algorithms are
summarized in Tables~\ref{tab:csc-results-int32-1},
\ref{tab:csc-results-int32-2}, \ref{tab:csc-results-int64-1}, and
\ref{tab:csc-results-int64-2} in comparison to those based on the floating-point
arithmetic with various precisions. We set $D$ to \num{1e6} and \num{1e8} for
32-bit integers and \num{1e9} and \num{1e12} for 64-bit integers, respectively,
and generate one million samples of $A$ corresponding to clock skew uniformly
distributed in the range of $[{-}100\,\text{ppm},100\,\text{ppm}]$. Given the
clock resolutions of \SI{1}{\us} and \SI{1}{\ns} for the two scenarios, the
minimum and the maximum values of $i$ (i.e., the hardware clock) correspond to
\SI{1}{\s} and \SI{1e3}{\s} in Table~\ref{tab:csc-results-int32-1},
\SI{1e-3}{\s} and \SI{1e1}{\s} in Table~\ref{tab:csc-results-int32-2},
\SI{1e3}{\s} and \SI{1e9}{\s} in Table~\ref{tab:csc-results-int64-1}, and
\SI{1e-2}{\s} and \SI{1e5}{\s} in Table~\ref{tab:csc-results-int64-2}
respectively. For the clock skew compensation based on the floating-point
arithmetic, we round half up the results to obtain the nearest integer values
throughout the experiments.
\begin{table*}[!tb]
  \centering
  \begin{threeparttable}
    \caption{Clock skew compensation based on integer linear scaling rounded to
      the nearest integer with 32-bit integers and
      $D{=}\num[round-precision=1]{1e6}$.}
    \label{tab:csc-results-int32-1} {%
      \sisetup{scientific-notation=true,round-mode=figures,group-digits=false}%
      \begin{tabular}{@{}
        l
        *{2}{S[round-precision=1,table-format=1e-1]}
        *{3}{S[round-precision=5,table-format=1.4e-1]}
        c
        @{}}%
        \toprule%
        \multicolumn{1}{c}{\multirow{2}{*}{Algorithm}} & \multicolumn{1}{c}{\multirow{2}{*}{$i$}} & \multicolumn{1}{c}{\multirow{2}{*}{$N$}} & \multicolumn{3}{c}{Compensation error\tnote{*}} & \multicolumn{1}{c}{\multirow{2}{*}{Notes}} \\
        \cmidrule{4-6}%
                                   & & & \multicolumn{1}{c}{Min.} & \multicolumn{1}{c}{Max.} & \multicolumn{1}{c}{Avg.} &  \\
        \midrule%
        \multirow{4}{*}{Double precision (binary64)\tnote{\dag}}
                                   & 1e6 & {--} & 0 & 0 & 0 & {--} \\
                                   & 1e7 & {--} & 0 & 0 & 0 & {--} \\
                                   & 1e8 & {--} & 0 & 0 & 0 & {--} \\
                                   & 1e9 & {--} & 0 & 0 & 0 & {--} \\
        \midrule%
        \multirow{4}{*}{Single precision (binary32)\tnote{\dag}}
                                   & 1e6 & {--} & 0 & 0 & 0 & {--} \\
                                   & 1e7 & {--} & 0 & 0 & 0 & {--} \\
                                   & 1e8 & {--} & -4 & 1 & -1.6962 & {--} \\
                                   & 1e9 & {--} & -19 & 44 & 1.2973e+01 & {--} \\
        \midrule%
        \multirow{4}{*}{\renewcommand{\arraystretch}{1.3}\begin{tabular}{@{}l@{}}Multiplicative decomposition of\\integer division (Theorem~\ref{thm:mdid})\end{tabular}} & 1e6 & {--} & {--} & {--} & {--} & Overflow \\ 
                                   & 1e7 & {--} & {--} & {--} & {--} & Overflow \\ 
                                   & 1e8 & {--} & {--} & {--} & {--} & Overflow \\ 
                                   & 1e9 & {--} & {--} & {--} & {--} & Overflow \\ 
        \midrule%
        \multirow{4}{*}{\renewcommand{\arraystretch}{1.3}\begin{tabular}{@{}l@{}}Additive decomposition of\\direct search (Theorem~\ref{thm:ds-dcmp})\end{tabular}} & 1e6 & 1 & 0 & 0 & 0 & -- \\
                                   & 1e7 & 1 & 0 & 0 & 0 & -- \\
                                   & 1e8 & 1e1 & 0 & 0 & 0 & -- \\
                                   & 1e9 & 100 & 0 & 0 & 0 & -- \\
        \bottomrule%
      \end{tabular}
    }%
    \begin{tablenotes}
    \item[*] With respect to ${\lfloor}\fp(i\frac{D}{A}){+}0.5{\rfloor}$ based
      on binary512 of IEEE 754-2008~\cite[Table~III]{Kang:24}.
    \item[\dag] ${\lfloor}\fp(i\frac{D}{A}){+}0.5{\rfloor}$ based on the
      corresponding precision.
    \end{tablenotes}
  \end{threeparttable}
\end{table*}

\begin{table*}[!tb]
  \centering
  \begin{threeparttable}
    \caption{Clock skew compensation based on integer linear scaling rounded to
      the nearest integer with 32-bit integers and
      $D{=}\num[round-precision=1]{1e8}$.}
    \label{tab:csc-results-int32-2} {%
      \sisetup{scientific-notation=true,round-mode=figures,group-digits=false}%
      \begin{tabular}{@{}
        l
        *{2}{S[round-precision=1,table-format=1e-1]}
        *{3}{S[round-precision=5,table-format=1.4e-1]}
        c
        @{}}%
        \hline%
        \multicolumn{1}{c}{\multirow{2}{*}{Algorithm}} & \multicolumn{1}{c}{\multirow{2}{*}{$i$}} & \multicolumn{1}{c}{\multirow{2}{*}{$N$}} & \multicolumn{3}{c}{Compensation error\tnote{*}} & \multicolumn{1}{c}{\multirow{2}{*}{Notes}} \\
        \cline{4-6}%
                                   & & & \multicolumn{1}{c}{Min.} & \multicolumn{1}{c}{Max.} & \multicolumn{1}{c}{Avg.} &  \\
        \hline%
        \multirow{5}{*}{Double precision (binary64)\tnote{\dag}}
                                   & 1e5 & {--} & 0 & 0 & 0 & {--} \\
                                   & 1e6 & {--} & 0 & 0 & 0 & {--} \\
                                   & 1e7 & {--} & 0 & 0 & 0 & {--} \\
                                   & 1e8 & {--} & 0 & 0 & 0 & {--} \\
                                   & 1e9 & {--} & 0 & 0 & 0 & {--} \\
        \hline%
        \multirow{5}{*}{Single precision (binary32)\tnote{\dag}}
                                   & 1e5 & {--} & 0 & 1 & 4.4958e-02 & {--} \\
                                   & 1e6 & {--} & 0 & 0 & 0 & {--} \\
                                   & 1e7 & {--} & 0 & 1 & 1.19e-01 & {--} \\
                                   & 1e8 & {--} & -4 & 5 & 1.8586e-01 & {--} \\
                                   & 1e9 & {--} & -40 & 82 & 1.4780e+01 & {--} \\
        \hline%
        \multirow{5}{*}{\renewcommand{\arraystretch}{1.3}\begin{tabular}{@{}l@{}}Multiplicative decomposition of\\fixed-point division (Theorem~\ref{thm:mdid})\end{tabular}} & 1e5 & {--} & {--} & {--} & {--} & Overflow \\ 
                                   & 1e6 & {--} & {--} & {--} & {--} & Overflow \\ 
                                   & 1e7 & {--} & {--} & {--} & {--} & Overflow \\ 
                                   & 1e8 & {--} & {--} & {--} & {--} & Overflow \\ 
                                   & 1e9 & {--} & {--} & {--} & {--} & Overflow \\ 
        \hline%
        \multirow{5}{*}{\renewcommand{\arraystretch}{1.3}\begin{tabular}{@{}l@{}}Additive decomposition of\\direct search (Theorem~\ref{thm:ds-dcmp})\end{tabular}} & 1e5 & 1 & 0 & 0 & 0 & {--} \\
                                   & 1e6 & 1e1 & 0 & 0 & 0 & {--} \\
                                   & 1e7 & 1e2 & 0 & 0 & 0 & {--} \\
                                   & 1e8 & 1e3 & 0 & 0 & 0 & {--} \\
                                   & 1e9 & 1e4 & 0 & 0 & 0 & {--} \\
        \hline%
      \end{tabular}
    }%
    \begin{tablenotes}
    \item[*] With respect to ${\lfloor}\fp(i\frac{D}{A}){+}0.5{\rfloor}$ based
      on binary512 of IEEE 754-2008~\cite[Table~III]{Kang:23}.
    \item[\dag] ${\lfloor}\fp(i\frac{D}{A}){+}0.5{\rfloor}$ based on single
      precision.
    \end{tablenotes}
  \end{threeparttable}
\end{table*}
%

\begin{table*}[!tb]
  \centering
  \begin{threeparttable}
    \caption{Clock skew compensation based on integer linear scaling rounded to
      the nearest integer with 64-bit integers and
      $D{=}\num[round-precision=1]{1e9}$.}
    \label{tab:csc-results-int64-1} {%
      \sisetup{scientific-notation=true,round-mode=figures,group-digits=false}%
      \begin{tabular}{@{}
        l
        *{2}{S[round-precision=1,table-format=1e-1]}
        *{3}{S[round-precision=5,table-format=1.4e-1]}
        @{}}%
        \toprule%
        \multicolumn{1}{c}{\multirow{2}{*}{Algorithm}} & \multicolumn{1}{c}{\multirow{2}{*}{$i$}} & \multicolumn{1}{c}{\multirow{2}{*}{$N$}} & \multicolumn{3}{c}{Compensation error\tnote{*}} \\
        \cmidrule{4-6}%
                                   & & & \multicolumn{1}{c}{Min.} & \multicolumn{1}{c}{Max.} & \multicolumn{1}{c}{Avg.} \\
        \midrule%
        \multirow{7}{*}{Quadruple precision (binary128)\tnote{\dag}}
                                   & 1e12 & {--} & 0 & 0 & 0 \\
                                   & 1e13 & {--} & 0 & 0 & 0 \\
                                   & 1e14 & {--} & 0 & 0 & 0 \\
                                   & 1e15 & {--} & 0 & 0 & 0 \\
                                   & 1e16 & {--} & 0 & 0 & 0 \\
                                   & 1e17 & {--} & 0 & 0 & 0 \\
                                   & 1e18 & {--} & 0 & 0 & 0 \\
        \midrule%
        \multirow{7}{*}{Double precision (binary64)\tnote{\dag}}
                                   & 1e12 & {--} & 0 & 0 & 0 \\
                                   & 1e13 & {--} & 0 & 0 & 0 \\
                                   & 1e14 & {--} & 0 & 0 & 0 \\
                                   & 1e15 & {--} & -1 & 0 & -3.9930e-02 \\
                                   & 1e16 & {--} & -2 & 0 & -8.7624e-01 \\
                                   & 1e17 & {--} & -13 & 3 & -4.9326 \\
                                   & 1e18 & {--} & -77 & 50 & -1.5377e+01 \\
        \midrule%
        \multirow{7}{*}{Single precision (binary32)\tnote{\dag}}
                                   & 1e12 & {--} & -7.9880e4 & 3.8672e4 & -2.0476e+04 \\
                                   & 1e13 & {--} & -5.8453e5 & 1.0421e6 & 2.1766e+05 \\
                                   & 1e14 & {--} & -5.1420e6 & 9.4285e6 & 2.0610e+06 \\
                                   & 1e15 & {--} & -7.5579e7 & 4.3954e7 & -1.4900e+07 \\
                                   & 1e16 & {--} & -1.5611e9 & 1.3190e8 & -5.9591e+08 \\
                                   & 1e17 & {--} & -9.8398e9 & 4.9293e9 & -2.5240e+09 \\
                                   & 1e18 & {--} & -5.0991e10 & 7.6703e10 & 1.1032e+10 \\
        \midrule%
        \multirow{7}{*}{\renewcommand{\arraystretch}{1.3}\begin{tabular}{@{}l@{}}Multiplicative decomposition of\\integer division (Theorem~\ref{thm:mdid})\end{tabular}} & 1e12 & {--} & 0 & 0 & 0 \\
                                   & 1e13 & {--} & 0 & 0 & 0 \\
                                   & 1e14 & {--} & 0 & 0 & 0 \\
                                   & 1e15 & {--} & 0 & 0 & 0 \\
                                   & 1e16 & {--} & 0 & 0 & 0 \\
                                   & 1e17 & {--} & 0 & 0 & 0 \\
                                   & 1e18 & {--} & 0 & 0 & 0 \\
        \midrule%
        \multirow{7}{*}{\renewcommand{\arraystretch}{1.3}\begin{tabular}{@{}l@{}}Additive decomposition of\\direct search (Theorem~\ref{thm:ds-dcmp})\end{tabular}} &  1e12 & 1 & 0 & 0 & 0 \\
                                   & 1e13 & 1 & 0 & 0 & 0 \\
                                   & 1e14 & 1e1 & 0 & 0 & 0 \\
                                   & 1e15 & 1e2 & 0 & 0 & 0 \\
                                   & 1e16 & 1e3 & 0 & 0 & 0 \\
                                   & 1e17 & 1e4 & 0 & 0 & 0 \\
                                   & 1e18 & 1e5 & 0 & 0 & 0 \\
        \bottomrule%
      \end{tabular}
    }%
    \begin{tablenotes}
    \item[*] With respect to ${\lfloor}\fp(i\frac{D}{A}){+}0.5{\rfloor}$ based
      on binary512 of IEEE 754-2008~\cite[Table~III]{Kang:24}.
    \item[\dag] ${\lfloor}\fp(i\frac{D}{A}){+}0.5{\rfloor}$ based on the
      corresponding precision.
    \end{tablenotes}
  \end{threeparttable}
\end{table*}

\begin{table*}[!tb]
  \centering
  \begin{threeparttable}
    \caption{Clock skew compensation based on integer linear scaling rounded to
      the nearest integer with 64-bit integers and
      $D{=}\num[round-precision=1]{1e12}$.}
    \label{tab:csc-results-int64-2} {%
      \sisetup{scientific-notation=true,round-mode=figures,group-digits=false}%
      \begin{tabular}{@{}
        l
        *{2}{S[round-precision=1,table-format=1e-1]}
        *{3}{S[round-precision=5,table-format=1.4e-1]}
        c
        @{}}%
        \midrule%
        \multirow{2}{*}{Algorithm} & \multicolumn{1}{c}{\multirow{2}{*}{$i$}} & \multicolumn{1}{c}{\multirow{2}{*}{$N$}} & \multicolumn{3}{c}{Compensation error\tnote{*}} & \multicolumn{1}{c}{\multirow{2}{*}{Notes}} \\
        \cmidrule{4-6}%
                                   & & & \multicolumn{1}{c}{Min.} & \multicolumn{1}{c}{Max.} & \multicolumn{1}{c}{Avg.} & \\
        \midrule%
        \multirow{9}{*}{Quadruple precision (binary128)\tnote{\dag}}
                                   & 1e10 & {--} & 0 & 0 & 0 & {--} \\
                                   & 1e11 & {--} & 0 & 0 & 0 & {--} \\
                                   & 1e12 & {--} & 0 & 0 & 0 & {--} \\
                                   & 1e13 & {--} & 0 & 0 & 0 & {--} \\
                                   & 1e14 & {--} & 0 & 0 & 0 & {--} \\
                                   & 1e15 & {--} & 0 & 0 & 0 & {--} \\
                                   & 1e16 & {--} & 0 & 0 & 0 & {--} \\
                                   & 1e17 & {--} & 0 & 0 & 0 & {--} \\
                                   & 1e18 & {--} & 0 & 0 & 0 & {--} \\
        \midrule%
        \multirow{9}{*}{Double precision (binary64)\tnote{\dag}}
                                   & 1e10 & {--} & 0 & 0 & 0 & {--} \\
                                   & 1e11 & {--} & 0 & 0 & 0 & {--} \\
                                   & 1e12 & {--} & 0 & 0 & 0 & {--} \\
                                   & 1e13 & {--} & -1 & 0 & -5.0650e-03 & {--} \\
                                   & 1e14 & {--} & -1 & 0 & -4.9660e-03 & {--} \\
                                   & 1e15 & {--} & -1 & 0 & -5.4861e-02 & {--} \\
                                   & 1e16 & {--} & -1 & 1 & 4.0231e-01 & {--} \\
                                   & 1e17 & {--} & 1 & 17 & 9.1050e+00 & {--} \\
                                   & 1e18 & {--} &  -83 & 44 & -1.9741e+01 & {--} \\
        \midrule%
        \multirow{9}{*}{Single precision (binary32)\tnote{\dag}}
                                   & 1e10 & {--} & -540 &  996 & 2.2672e+02 & {--} \\
                                   & 1e11 & {--} & -3804 & 8973 & 2.4693e+03 & {--} \\
                                   & 1e12 & {--} & -23943 & 48410 & 1.1400e+04 & {--} \\
                                   & 1e13 & {--} & -361943 & 1191550 & 4.3032e+05 & {--} \\
                                   & 1e14 & {--} & -8863718 & 4035600 & -2.4610e+06 & {--} \\
                                   & 1e15 & {--} & -48584545 & 73491195 & 1.1034e+07 & {--} \\
                                   & 1e16 & {--} & -1421855366 & 128570432 & -6.2788e+08 & {--} \\
                                   & 1e17 & {--} & -7943826631 & 5656767492 & -1.4239e+09 & {--} \\
                                   & 1e18 & {--} &  -79025477331 & 49332003960 & -1.5484e+10 & {--} \\
        \midrule%
        \multirow{9}{*}{\renewcommand{\arraystretch}{1.3}\begin{tabular}{@{}l@{}}Multiplicative decomposition of\\fixed-point division (Theorem~\ref{thm:mdid})\end{tabular}} & 1e10 & {--} & {--} & {--} & {--} & Overflow \\
                                   & 1e11 & {--} & {--} & {--} & {--} & Overflow \\
                                   & 1e12 & {--} & {--} & {--} & {--} & Overflow \\
                                   & 1e13 & {--} & {--} & {--} & {--} & Overflow \\
                                   & 1e14 & {--} & {--} & {--} & {--} & Overflow \\
                                   & 1e15 & {--} & {--} & {--} & {--} & Overflow \\
                                   & 1e16 & {--} & {--} & {--} & {--} & Overflow \\
                                   & 1e17 & {--} & {--} & {--} & {--} & Overflow \\
                                   & 1e18 & {--} & {--} & {--} & {--} & Overflow \\
        \midrule%
        \multirow{9}{*}{\renewcommand{\arraystretch}{1.3}\begin{tabular}{@{}l@{}}Additive decomposition of\\direct search (Theorem~\ref{thm:ds-dcmp})\end{tabular}} &  1e10 & 1 & 0 & 0 & 0 \\
                                   & 1e11 & 1e1 & 0 & 0 & 0 & {--} \\
                                   & 1e12 & 1e2 & 0 & 0 & 0 & {--} \\
                                   & 1e13 & 1e3 & 0 & 0 & 0 & {--} \\
                                   & 1e14 & 1e4 & 0 & 0 & 0 & {--} \\
                                   & 1e15 & 1e5 & 0 & 0 & 0 & {--} \\
                                   & 1e16 & 1e6 & 0 & 0 & 0 & {--} \\
                                   & 1e17 & 1e7 & 0 & 0 & 0 & {--} \\
                                   & 1e18 & 1e8 & 0 & 0 & 0 & {--} \\
        \bottomrule%
      \end{tabular}
    }%
    \begin{tablenotes}
    \item[*] With respect to ${\lfloor}\fp(i\frac{D}{A}){+}0.5{\rfloor}$ based
      on binary512 of IEEE 754-2008~\cite[Table~III]{Kang:23}.
    \item[\dag] ${\lfloor}\fp(i\frac{D}{A}){+}0.5{\rfloor}$ based on the
      corresponding precision.
    \end{tablenotes}
  \end{threeparttable}
\end{table*}
%


Regarding the overflows, we observe that the multiplicative decomposition of
integer division algorithm overflows for all the cases under the first scenario
where the ratio of $D$ to the maximum value of the underlying integer type is
relatively larger (i.e., $\frac{D}{2^{31}-1}{\approx}\num{4.6566e-4}$), while it
can handle all the cases of the second scenario when the said ratio is much
smaller (i.e., $\frac{D}{2^{63}-1}{\approx}\num{1.0842e-10}$). In comparison,
the additive decomposition of direct search algorithm can handle all the cases
under both scenarios without overflows but at the expense of increased
computational complexity when the value of $i$ approaches the maximum value of
the underlying integer type.

As for the compensation errors, we observe that the additive decomposition of
direct search algorithm is equivalent to the clock skew compensation based on
64-bit double-precision floating-point arithmetic under the first scenario based
on 32-bit integers; under the second scenario based on 64-bit integers, on the
other hand, both algorithms are equivalent to the clock skew compensation based
on 128-bit quadruple-precision floating-point arithmetic. This is remarkable in
that the decomposition algorithms based only on integer operations require only
half the size of the data types for the clock skew compensation based on
floating-point arithmetic for the same compensation error.

\section{Conclusions}
\label{sec:conclusions}
We have generalized the problem of clock skew compensation as the integer linear
scaling problem in the form of $i\frac{D}{A}$ and proposed the multiplicative
decomposition of integer division and the additive decomposition of direct
search algorithms, which decompose the product of $iD$ and the input $i$,
respectively, during the intermediate calculation to reduce overflows in
obtaining the nearest integer solution based on a fixed-width integer type. The
proposed algorithms are not only immune to floating-point precision loss due to
their using only integer operations but also non-incremental to reduce their
computational complexities unlike our prior approaches based on Bresenham's
algorithm.

We have theoretically established both decomposition algorithms based on a
unified and rigorous formulation of the problem of the integer linear scaling
rounded to the nearest integer---which has been studied in several disciplines
under different names but not based on a common framework and, oftentimes,
without formal proofs and taking into account rounding errors---and discussed
the space-time trade-off through the analysis of their computational
complexities and non-overflow conditions.

Through numerical examples, we have demonstrated the relative advantages and
disadvantages of the two decomposition algorithms in the practical context of
clock skew compensation under two different scenarios. Specifically, the
multiplicative decomposition of integer division algorithm overflows under the
first scenario based on 32-bit integers, while it can handle all the cases of
the second scenario based on 64-bit integers with a constant time complexity
(i.e., $\bigO(1)$); in comparison, the additive decomposition of direct search
algorithm can handle all the cases under both scenarios without overflows but at
the expense of increased computational complexity when the value of $i$
approaches the maximum value of the underlying integer type. The results also
highlight the trade-off between the bounded compensation errors and lower space
complexity of the integer-based decomposition algorithms and the lower chances
of overflows resulting from the wide ranges of numbers of the clock skew
compensation based on floating-point arithmetic.

Of the two algorithms, the results of the analysis of computational complexities
and non-overflow conditions show that the direct search algorithm can better
exploit the trade-off between space and time complexity in obtaining the nearest
integer solution to integer linear scaling through additive decomposition. Given
the sequential nature of $\tilde{\Delta}_{n}$ in \eqref{eq:thm-ds-dcmp-2} of
Theorem~\ref{thm:ds-dcmp}, however, the parallelization of the additive
decomposition of direct search algorithm is not straightforward, which is an
interesting topic for further research together with an alternative setting of
an initial guess of the solution (i.e., $k_{0}$ in Algorithm~\ref{alg:ds}) to
avoid an overflow given the values of $i$, $D$, and $A$.

\appendices

\section{Proof of Lemma~\ref{lem:ils-rttn-fp}}
\label{sec:lem-ils-rttn-fp-proof}
Let $a$ and $b$ denote the integer and the fractional part of $i\frac{D}{A}$
such that $i\frac{D}{A}{=}a{+}b$ where $a{\in}\mathbb{N}_{0}$,
$b{\in}\mathbb{R}$, and $0{\leq}b{<}1$.

Suppose that there exists $k{\in}\mathbb{N}_{0}$ such that $k{=}j{+}l$, where
$l{\in}\mathbb{N}$ and $|l|{\geq}1$, and
\begin{equation}
  \label{eq:ils-rttn-fp-assumption}
  \left| j - i\frac{D}{A} \right| > \left| k - i\frac{D}{A} \right|.
\end{equation}
  
If $b{<}0.5$, $j$ becomes $a$. Also,
\[
  \left| j - i\frac{D}{A} \right| = \left| a - (a + b) \right| = \left| b
  \right| < 0.5,
\]
and
\[
  \left| k - i\frac{D}{A} \right| = \left| (a + l) - (a + b) \right| = \left| l
    - b \right| > 0.5.
\]
Therefore,
\begin{equation}
  \label{eq:ils-rttn-fp-contradiction1}
  \left| j - i\frac{D}{A} \right| < \left| k - i\frac{D}{A} \right|.
\end{equation}

If $b{\geq}0.5$, on the other hand, $j$ becomes $a{+}1$. Also,
\[
  \left| j - i\frac{D}{A} \right| = \left| (a + 1) - (a + b) \right| = \left| 1
    - b \right| \leq 0.5,
\]
and
\[
  \left| k - i\frac{D}{A} \right| = \left| (a + 1 + l) - (a + b) \right| =
  \left| 1 + l - b \right| \geq 0.5.
\]
Therefore,
\begin{equation}
  \label{eq:ils-rttn-fp-contradiction2}
  \left| j - i\frac{D}{A} \right| \leq \left| k - i\frac{D}{A} \right|.
\end{equation}
\eqref{eq:ils-rttn-fp-contradiction1} and \eqref{eq:ils-rttn-fp-contradiction2}
contradict the assumption of the existence of $k$ satisfying
\eqref{eq:ils-rttn-fp-assumption}. Therefore, $j$ satisfies
\eqref{eq:ils-rttn-1} and is the nearest integer solution to $i\frac{D}{A}$.
\hfill\IEEEQED

\section{Proof of Theorem~\ref{thm:mdid}}
\label{sec:thm-mdid-proof}
According to the theorem of division algorithm (also called Euclid's division
lemma)~\cite[Theorem~2.1]{burton10:_elemen_number_theor}, there exist unique
integers $q,r{\in}\mathbb{N}_{0}$ such that
\begin{equation}
  \label{eq:mdid-proof-1}
  i = Aq + r,
\end{equation}
where $q{=}{\lfloor}\frac{i}{A}{\rfloor}$ and $r{=}i \bmod A$. Multiplying both
sides of \eqref{eq:mdid-proof-1} by $\frac{D}{A}$, we obtain
\begin{equation}
  \label{eq:mdid-proof-2}
  i\frac{D}{A} = qD + \frac{rD}{A} = \left\lfloor\frac{i}{A}\right\rfloor D +
  \frac{(i \bmod A)D}{A}.
\end{equation}
Applying \eqref{eq:mdid-proof-2} to \eqref{eq:ils-rttn-fp} of
Lemma~\ref{lem:ils-rttn-fp}, the nearest integer solution $j$ is given by
\begin{equation}
  \label{eq:mdid-proof-3}
  \begin{split}
    j & = \left\lfloor\left(\left\lfloor\frac{i}{A}\right\rfloor D + \frac{(i \bmod A)D}{A}\right) + 0.5 \right\rfloor \\
      & = \left\lfloor\frac{i}{A}\right\rfloor D + \left\lfloor\frac{(i \bmod A)D + \frac{A}{2}}{A}\right\rfloor \\
      & = \left\lfloor\frac{i}{A}\right\rfloor D + \left\lfloor\frac{\left(i \bmod A\right)D + \left\lfloor\frac{A}{2}\right\rfloor}{A}\right\rfloor,
  \end{split}
\end{equation}
where the last step (i.e., $\frac{A}{2}{\rightarrow}\lfloor\frac{A}{2}\rfloor$
in the numerator) is from the special property of the floor
function~\cite[Eq.~(3.11)]{graham94:_concr}.
\hfill\IEEEQED

\section{Proof of Theorem~\ref{thm:ds}}
\label{sec:thm-ds-proof}
Given the objective function $f(x){=}|x{-}i\frac{D}{A}|$, we can decrease its
value by decreasing $x$ as far as $x{>}i\frac{D}{A}$ or increasing $x$ as far as
$x{<}i\frac{D}{A}$. Because $x$ is constrained to be an integer, there could be
at most two integer values of $x$ minimizing $f(x)$ around
$x{=}i\frac{D}{A}$. Therefore, if we start the minimization from either side of
$x{=}i\frac{D}{A}$ and cannot decrease any further the value of $f(x)$ by
decreasing (i.e., $x{>}i\frac{D}{A}$) or increasing (i.e., $x{<}i\frac{D}{A}$)
$x$, it means that we have reached the value of $x$ minimizing $f(x)$, which is
a termination criterion for Algorithm~\ref{alg:ds}.

Having the termination criterion in mind, therefore, we consider each case of
Algorithm~\ref{alg:ds}:
  
\mycase{1}: This is a special case of Lemma~\ref{lem:ils-rttn}, where rhs of
\eqref{eq:ils-rttn-2} is zero. Because
\[
  \Delta_{0} = \left(k_{0}A - iD\right) = 0,
\]
we have from Lemma~\ref{lem:ils-rttn}
\[
  \left|jA-iD\right| \leq \left|k_{0}A-iD\right| = 0,
\]
which results in $j{=}k_{0}$.

\medskip%

When $\Delta_{0}{\neq}0$,
there exist unique integers $q,r{\in}\mathbb{N}_{0}$ such that
\begin{equation}
  |\Delta_{0}| = Aq + r,
\end{equation}
where $q{=}{\lfloor}\frac{|\Delta_{0}|}{A}{\rfloor}$ and
$r{=}|\Delta_{0}| \bmod A$ according to the division algorithm.

\mycase{2}: Because $\Delta_{0}{>}0$, we have
\begin{equation}
  k_{1} = k_{0} - \left\lfloor \dfrac{\Delta_{0}}{A} \right\rfloor = k_{0} - q,
\end{equation}
\begin{equation}
  \begin{split}
    \Delta_{1} & = \Delta_{0} \bmod A = r = \Delta_{0} - Aq \\
            & = \left(k_{0}A - iD\right) - Aq = (k_{0}-q)A - iD \\
            & = k_{1}A - iD.
  \end{split}
\end{equation}
Because $r{=}\Delta_{1}{<}A$, we could obtain the minimum value of rhs of
\eqref{eq:ils-rttn-2} by comparing $|\Delta_{1}|$ and $|\Delta_{1}{-}A|$.

If $|\Delta_{1}{-}A|{<}|\Delta_{1}|$, it suffices to show that
$|\Delta_{1}{-}A|$ is the rhs of \eqref{eq:ils-rttn-2} corresponding to
$k_{1}{-}1$, that is,
\[
  |\Delta_{1} - A| = \left|(k_{1}A - iD) - A\right| = \left|(k_{1} - 1)A -
    iD\right|.
\]
Because $|\Delta_{1}{-}A|$ is the minimum value of rhs of \eqref{eq:ils-rttn-2},
$k_{1}{-}1$ is $j$ according to Lemma~\ref{lem:ils-rttn}
(\textit{Case}~2.1). Otherwise, $k_{1}$ is $j$ (\textit{Case}~2.2).

\mycase{3}: Because $\Delta_{0}{<}0$, we have
\begin{equation}
  k_{1} = k_{0} - \left\lfloor \dfrac{\Delta_{0}}{A} \right\rfloor = k_{0} + q,
\end{equation}
\begin{equation}
  \begin{split}
    \Delta_{1} & = -(|\Delta_{0}| \bmod A) = -r = \Delta_{0} + Aq \\
            & = (k_{0}A - iD) + Aq = (k_{0}+q)A - iD \\
            & = k_{1}A - iD.
  \end{split}
\end{equation}
Because $r{=}{-}\Delta_{1}{<}A$, we could obtain the minimum value of rhs of
\eqref{eq:ils-rttn-2} by comparing $|\Delta_{1}|$ and $|\Delta_{1}{+}A|$.

If $|\Delta_{1}{+}A|{<}|\Delta_{1}|$, it suffices to show that
$|\Delta_{1}{+}A|$ is the rhs of \eqref{eq:ils-rttn-2} corresponding to
$k_{1}{+}1$, that is,
\[
  |\Delta_{1} + A| = \left|(k_{1}A - iD) + A\right| = \left|(k_{1} + 1)A -
    iD\right|.
\]
Because $|\Delta_{1}{+}A|$ is the minimum value of rhs of \eqref{eq:ils-rttn-2},
$k_{1}{+}1$ is $j$ according to Lemma~\ref{lem:ils-rttn}
(\textit{Case}~3.1). Otherwise, $k_{1}$ is $j$ (\textit{Case}~3.2).
\hfill\IEEEQED

\section{Proof of Lemma~\ref{lem:ds-dcmp}}
\label{sec:lem-ds-dcmp-proof}
If suffices to show that
$ds(i_{2},D,A,\kappa_{2},\tilde{\Delta}_{1}){=}j{-}j_{1}$.

In fact, $\Delta_{0}$ in line 2 of $ds(i_{2},D,A,\kappa_{2},\tilde{\Delta}_{1})$
is identical to that of $ds(i,D,A,\kappa_{2}{+}j_{1},0)$; because
$\tilde{\Delta}_{1}{=}j_{1}A{-}i_{1}D$, $\Delta_{0}$ in line 2 is given by
\begin{equation}
  \label{eq:lem-ds-dcmp-proof-1}
  \begin{split}
    \Delta_{0} & = (\kappa_{2} - i_{2})A + i_{2}(A - D) + \tilde{\Delta}_{1} \\ 
               & = (\kappa_{2}A - i_{2}D) + \tilde{\Delta}_{1} \\
               & = (\kappa_{2}A - i_{2}D) + (j_{1}A - i_{1}D) \\
               & = (\kappa_{2} + j_{1})A - (i_{1} + i_{2})D \\
               &  = (\kappa_{2} + j_{1})A - iD.
  \end{split}
\end{equation}

The only internal difference between two function calls
$ds(i_{2},D,A,\kappa_{2},\tilde{\Delta}_{1})$ and
$ds(i,D,A,\kappa_{2}{+}j_{1},0)$ is the value of $k_{0}$ (i.e., $\kappa_{2}$
vs. $\kappa_{2}{+}j_{1}$) from line 4 and on, which determines the return value
directly (i.e., line 4) or indirectly through $k_{1}$ (i.e., lines 9, 12, 18,
and 21). As the value of $k_{0}$ does not affect the conditional processing,
which depends only on the value of $\Delta_{0}$ and $\Delta_{1}$, the return
value between the two function calls differs by $j_{1}$. As
$ds(i,D,A,\kappa_{2}{+}j_{1},0)$ returns $j$,
$ds(i_{2},D,A,\kappa_{2},\Delta_{1})$ returns $j{-}j_{1}$.
\hfill\IEEEQED

\section{Proof of Theorem~\ref{thm:ds-dcmp}}
\label{sec:thm-ds-dcmp-proof}
The base case of $N{=}1$ holds by Theorem~\ref{thm:ds}.

Assume that \eqref{eq:thm-ds-dcmp-1} and \eqref{eq:thm-ds-dcmp-2} give
$\check{j}_{l}{=}\sum_{n=1}^{l}j_{n}$, the nearest integer solution to
$\check{i}_{l}\frac{D}{A}$, where $\check{i}_{l}{=}\sum_{n=1}^{l}i_{n}$, and
$\check{\Delta}_{l}$, the value of $\check{j}_{l}A{-}\check{i}_{l}D$, for an
integer $l{\geq}1$. Based on Lemma~\ref{lem:ds-dcmp}, we can obtain
$\check{j}_{l+1}$, the nearest integer solution to $\check{i}_{l+1}\frac{D}{A}$,
where $\check{i}_{l+1}{=}\sum_{n=1}^{l+1}i_{n}$, and $\check{\Delta}_{l+1}$, the
value of $\check{j}_{l+1}A{-}\check{i}_{l+1}D$, as follows:
\begin{equation}
  \label{eq:thm-ds-dcmp-proof-1}
  \begin{split}
    \check{j}_{l+1} & = \check{j}_{l} + j_{l+1} = \sum_{n=1}^{l+1}j_{n}, \\
    \tilde{\Delta}_{l+1} & = \tilde{\Delta}_{l+1},
  \end{split}
\end{equation}
where
\begin{equation}
  \label{eq:thm-ds-dcmp-proof-2}
  \begin{split}
    (\check{j}_{l}, \tilde{\Delta}_{l}) & = ds(\check{i}_{l}, D, A, \check{\kappa}_{l}, 0), \\
    (j_{l+1}, \tilde{\Delta}_{l+1}) & = ds(i_{l+1}, D, A, \kappa_{l+1}, \tilde{\Delta}_{l}),
  \end{split}
\end{equation}
and $\check{\kappa}_{l}{=}\sum_{n=1}^{l}\kappa_{n}$ and
$\kappa_{l+1}{\in}\mathbb{N}_{0}$ are the initial guesses of $\check{j}_{l}$ and
$j_{l+1}$, respectively.

Therefore, \eqref{eq:thm-ds-dcmp-1} and \eqref{eq:thm-ds-dcmp-2} give $j$, the
nearest integer solution to $i\frac{D}{A}$, where $i{=}\sum_{n=1}^{N}i_{n}$, and
$\Delta{=}jA{-}iD$ for every integer $N{\geq}1$ by mathematical induction.
\hfill\IEEEQED

\balance 


\end{document}